\documentclass[11pt]{article} 

\usepackage[utf8]{inputenc} 

\usepackage{cancel}
\usepackage[section]{placeins}

\usepackage{geometry} 
\usepackage{graphicx} 
\usepackage[english]{babel}

\usepackage{booktabs} 
\usepackage{array} 
\usepackage{paralist} 
\usepackage{verbatim} 
\usepackage{subfig} 
\usepackage{amsfonts}
\usepackage{amssymb}
\usepackage{amsmath}
\usepackage{amsthm}
\usepackage{makecell}
\usepackage{multirow}

\usepackage{epigraph}
\usepackage{bbm}

\usepackage{hyperref}
\hypersetup{
    colorlinks=true,
    citecolor=blue
}

\usepackage{bbm}

\renewcommand{\le}{\leqslant}

\renewcommand{\ge}{\geqslant}

\DeclareMathOperator*{\argmax}{argmax}

\newcommand{\id}{\mathbbmss 1}

\newcommand {\matl}{\left[ \begin{matrix}}
\newcommand {\matr}{\end{matrix}\right]}
\newcommand {\Exp}{ \mathbb E }

\renewcommand {\Pr}{ \mathbb P }

\newcommand {\Var}{\mathbf {Var}}

\newcommand{\cX}{\mathcal{X}}
\newcommand{\cY}{\mathcal{Y}}

\newcommand{\cF}{\mathcal{F}}
\newcommand{\cG}{\mathcal{G}}

\newcommand{\cM}{\mathcal{M}}
\newcommand{\ST}{\mathbb{T}}
\newcommand{\cH}{\mathcal{H}}
\newcommand{\cP}{\mathcal{P}}

\newcommand{\iid}{\overset{\mathrm{iid}}{\sim}}

\DeclareMathAlphabet{\mathbbmsl}{U}{bbm}{m}{sl}

\DeclareMathOperator{\rad}{Rad}

\newcommand{\Y}{\mathbf{Y}}

\newcommand{\fdr}{\mathsf{fdr}}

\usepackage[round]{natbib}
\usepackage{wrapfig}
\usepackage{subfig}
\usepackage[capitalize,noabbrev]{cleveref}
\usepackage{authblk}
\usepackage{changepage}

\usepackage{algorithm}
\usepackage{algpseudocode}

\title{Anytime-valid FDR control with the stopped e-BH procedure\footnote{To appear in the special issue ``E-values and Multiple Testing'' of \emph{Statistics \& Probability Letters}}
}

\author[1]{Hongjian Wang}
\author[1]{Sanjit Dandapanthula}
\author[2]{Aaditya Ramdas}
\affil[1, 2]{Department of Statistics and Data Science, Carnegie Mellon University}
\affil[2]{Machine Learning Department, Carnegie Mellon University} 
\affil[ ]{\texttt{ \{hjnwang,sanjitd,aramdas\}@cmu.edu  }}

\date{\today}

\newtheorem{theorem}{Theorem}[section]
\newtheorem{definition}[theorem]{Definition}

\newtheorem{proposition}[theorem]{Proposition}

\newtheorem{corollary}[theorem]{Corollary}
\newtheorem{lemma}[theorem]{Lemma}
\newtheorem{remark}[theorem]{Remark}
\newtheorem{example}[theorem]{Example}
\newtheorem{assumption}[theorem]{Assumption}

\usepackage{tocbasic}
\DeclareTOCStyleEntry[
  beforeskip=.05em plus .6pt,
  pagenumberformat=\textbf
]{tocline}{section}

\begin{document}
\maketitle

\begin{abstract} The e-Benjamini-Hochberg (e-BH) procedure for multiple hypothesis testing is known to control the false discovery rate (FDR) under arbitrary dependence between the input e-values. This paper points out an important subtlety when applying e-BH to e-processes, the sequential counterparts of e-values: stopping multiple e-processes at a common stopping time only yields e-values if all the e-processes and the stopping time are with respect to the same global filtration.
    We show that this filtration issue is of real concern as e-processes are often constructed to be ``local'' as opposed to ``global''. We formulate a condition under which these local e-processes are indeed global and thus applying e-BH to their stopped values (the ``stopped e-BH procedure'') controls the FDR. The condition excludes
    confounding from the past and is met under most reasonable scenarios including genomics.
   \end{abstract}

\section{Introduction}


We consider the false discovery rate (FDR) control problem in sequential multiple hypothesis testing, where batches of data that corresponds to a fixed number of hypotheses arrive in sequence.
The Benjamini-Hochberg (BH) procedure \citep{benjamini1995controlling} is among the most widely used approaches to FDR control. However, the BH procedure operates directly on fixed-sample p-values, which are inherently non-sequential. Further, the BH procedure only controls the FDR under independence or certain assumptions on the ``positive correlation'' between hypotheses \citep{finner-prds}, and may fail if the hypotheses are arbitrarily correlated.

Recently, \cite{wang2022false} showed that, by applying BH to the reciprocals of \emph{e-values}, a procedure referred to as ``e-BH'', the FDR control holds under arbitrary dependence across hypotheses. Further, since e-values arise naturally in sequential experiments, e-BH opens up the possibility of sequential multiple testing. The sequential version of e-BH was later spelled out by \cite{xu2021unified}, whose result allows bandit-like active queries into the hypotheses, each carrying an e-process (to be defined formally later), and the FDR is controlled at any \emph{stopping time} under arbitrary dependence. 

This seemingly very satisfactory result by \cite{xu2021unified}, however, bears a very crucial caveat. It is assumed that the e-processes of the hypotheses are valid under a shared ``global'' filtration. This indicates that one can \emph{not},  without further verification, 
simply apply e-BH to e-processes each constructed ``locally'' \emph{within} the hypotheses at a shared stopping time.

In this paper, we first formally define the local-global distinction of e-processes, and identify some further conditions under which e-processes constructed within hypotheses using generic methods {are not only local, but also} global e-processes. 
Additionally, we define the more generalized notion of global compound e-processes, and show that e-BH with these processes is \emph{necessary and sufficient} to obtain stopped FDR control.

\section{Background and problem set-up}

Throughout the paper, we shall liberally adopt the notations $\Exp_P, \Exp_{\theta}$ etc.\ to denote the expected value of some random variable when the underlying data-generating distribution is $P$, or when the underlying parameter of interest takes the value of $\theta$.

\subsection{Review: multiple testing with e-values and e-BH}\label{sec:mt}

 A nonnegative random variable $E$ is called an e-value for the set of distributions $\cP$, if $\sup_{P\in\cP}\Exp_{P}E \le 1$. To test the null hypothesis ``the true distribution $P^* \in \cP$'', one constructs an e-value $E$ for $\cP$ and rejects if $E \ge 1/\alpha$. This controls the type 1 error rate due to Markov's inequality.

In the multiple testing set-up, {let $\cM$ be the}  set of all possible distributions, {and let} $\cP^1, \dots, \cP^G \subseteq \cM$ {be $G$ null sets}. {Let} $[G]$ {be} the set $\{ 
1,\dots, G \}$. {Let $P^*$ be the unknown true} distribution {which belongs to $\cM$. It}
may {belong to arbitrary number of null sets among $(\cP^g : g\in [G])$}.
The $g$\textsuperscript{th} null hypothesis $\cH_0^g$ is true if rd{$P^* \in \cP^g$}. For any random set $R \subseteq [G]$, we denote its false discovery rate (FDR) under {some $P \in \cM$} by 
\begin{equation}
   \fdr_P[R] := \Exp_P \left[\frac{\sum_{g=1}^G \id\{  P \in \cP^g  \} \cdot  \id\{ g \in R  \}   }{ 1 \vee \sum_{g=1}^G \id\{ g \in R  \}  } \right].
\end{equation}
{Let $E_{(h)}$ be the $h$\textsuperscript{th} largest elements among $(E_g : g\in [G])$. The e-BH procedure \citep{wang2022false} at level $\alpha$ is defined as}
\begin{equation*}
     \mathsf{eBH}_{\alpha}(E_g : g\in[G]) = \{ g: E_g \ge E_{(g^*)} \} \quad \text{where }
    g^* =  \max\left\{ g \in [G] : \frac{g E_{(g)}}{G} \ge \frac{1}{\alpha} \right\},
\end{equation*}
and satisfies \citep[Proposition 2]{wang2022false}:
\begin{lemma}[FDR control of e-BH]\label{lem:fdr-ebh} 
For each $g\in [G]$, suppose $E_g$ is an e-value for $\cP^g$. Then $\fdr_P[\mathsf{eBH}_{\alpha}(E_{g} : g \in [G]) ]  \le \alpha$ for any {$P \in \cM$}. 
\end{lemma}
A generalization with ``compound e-values'' and a converse to \cref{lem:fdr-ebh} are discussed in \ref{app:comp}. We also remark that, in \cref{lem:fdr-ebh} and its downstream statements, we can replace $\mathsf{eBH}_{\alpha}$ with the ``closed e-BH'' procedure recently proposed by \cite{xu2025bringing}, which improves the power while retaining the FDR control.

\subsection{Sequential multiple testing, filtrations, and the stopped e-BH}\label{sec:seq}

We now introduce the sequential experiment setting and let us first review the terminology with a single hypothesis. Let $ \{ \cF_n \}_{n \ge 0}$ be a filtration. An  $ \{ \cF_n \}$-adapted nonnegative process $\{ M_n \}_{n \ge 1}$ is called an e-process \emph{for $\cP$ on $\{ \cF_n \}$} if either of the following two equivalent conditions holds:
\begin{itemize}
    \item At any $\{ \cF_n \}$-stopping time $\tau$, $M_\tau$ is an e-value for $\cP$.
    \item For any $P \in \cP$, there is a process $\{ N^P_n \}$ such that $N^P_0 = 1$, it is a nonnegative supermartingale under $P$, and $N^P_n \ge M_n$ for all $n$.
\end{itemize}
Therefore, one can stop at any $\{ \cF_n \}$-stopping time $\tau$ and safely reject the null $\cP$ if the stopped e-value $M_\tau \ge 1/\alpha$. In particular, $\tau$ can be $\min\{n : M_n \ge 1/\alpha \}$. 

With \cref{lem:fdr-ebh}, it is tempting to consider the following sequential multiple testing scheme with e-processes: since ``the stopped value of an e-process is an e-value'', if one runs an e-process for each hypothesis and computes e-BH with the stopped values of e-processes, does one achieve the sequential e-BH control at any stopping time?
Here, however, lies the subtlety of the problem: \emph{filtrations}. Let us carefully define the sequential problem with multiple hypotheses. 

We assume $G$ null sets of distributions $\cP^1,\dots, \cP^G {\subseteq \cM}$ as before, and that the $g$\textsuperscript{th} null $\cH_0^g$ states that the underlying {$P^*$ belongs to $\cP^g$}.
Further, for each $g \in [G]$,
 there is a \emph{local} filtration $\{ \cF^g_n \}_{n\ge 0}$ generated from data for the hypothesis $\cH^g_0$. 
 The combination of these $G$ local filtrations gives rise to the \emph{global} filtration $\cF_n = \sigma( \cF_n^g : g \in [G] )$,
 representing all information gathered up to $n$ observations from all hypotheses. 
 To sequentially test each hypothesis $\cH^g_0$, one naturally utilizes data on $\{\cF^g_n\}$ and constructs an e-processes on it. Formally, we refer to an e-process or a stopping time on $\{\cF_n\}$ as a \emph{global} one; and on $\{ \cF_n^g \}$, a ($g$-)\emph{local} one.
 
 The following coin-toss example illustrates this local-global distinction, and the alerting fact that per-hypothesis e-processes are by default local, not global. We shall revisit this model later. Let $\rad(p)$ be the Rademacher distribution with mass $p$ and $1-p$ on $\pm 1$ respectively.
 \begin{example}[Multiple coin tosses]\label{ex:coin}
     For each $g \in [G]$, there is a stream $\{Y_n^g\} \iid \operatorname{Rad}(\theta^g)$ leading to a local filtration $\cF^g_n = \sigma( Y_i^g,\dots, Y_n^g )$. 
     The process $
         M_n^g = \prod_{i=1}^n \left( 1 + Y^g_i / 2 \right) $
     is an e-process on the local filtration $\{\cF^g_n\}$ for the null $\theta^g = 1/2$. Without further assumptions, there is \underline{no guanratee} that $ \{M_n^g \}$ is an e-process on the global filtration $\cF_n = \sigma(  \cF_n^g : g \in [G] )$.
 \end{example}

The fact that $\{M_n^g \}$ is a local e-process can be easily seen by noting $\Exp_{\theta^g=1/2}[1 + Y^g_i / 2  | \cF_{n-1}^g ] = 1$.
This simple example already alarms that it is only safe to stop these $G$ per-hypothesis e-processes at local stopping times. We shall later explicitly spell out an example where $G = 2$ and the local e-process $\{ M_n^1 \}$ is not an e-process on the global filtration $\{\cF_n\}$, in \ref{sec:counterexample}. We also discuss in \ref{app:ambient} a deceptively promising variation of \cref{ex:coin} by changing the ``ambient filtration'' in each local e-process.

The concept of local and global filtrations and e-processes leads to that of the stopped FDR control, which is only defined globally. {We denote the set of all stopping times on $\{\cF_n \}$ by $\ST$.}

\begin{definition}\label{def:stopped-fdr}
    A sequence of random sets $\{ R_n \}$ (each $R_n \subseteq [G]$) adapted to $\{\cF_{n}\}$ satisfies the \emph{level-$\alpha$ stopped FDR control} if
    $\sup_{ \substack{ P \in \cM \\ \tau \in \ST }}\fdr_P [R_{\tau}] \le \alpha$.
   
\end{definition}

Our definitions lead to the following procedure which we call the \emph{stopped e-BH}. Its FDR ganruantee is a
direct corollary of \cref{lem:fdr-ebh}.

\begin{theorem}[Stopped e-BH]\label{thm:stop-ebh}
    Let $(\{ M^g_n \}: g\in [G])$ be global e-processes. Then, the set process $\{ \mathsf{eBH}_{\alpha}(M^g_n: g\in [G]  )\}$  satisfies the {level-$\alpha$ stopped FDR control}.
\end{theorem}

We shall demonstrate in \ref{app:comp} a converse to \cref{thm:stop-ebh}: {a procedure controls the stopped FDR if and only if it is the stopped e-BH procedure with global ``compound'' e-processes}. It is worth remarking here that \cref{def:stopped-fdr} and \cref{thm:stop-ebh} arise since we desire stopping at a global stopping time $\tau \in \ST$. Indeed, it is correct that if one wishes to stop at a $g$-local stopping time $\tau^g$ for \emph{each} hypothesis $\cH^g_0$, one can work with local e-processes: the stopped local e-process values $(M^g_{\tau^g} : g \in [G])$ are e-values to which one may safely apply e-BH.  This, however, disallows stopping any stream after observing the \emph{output} of e-BH, e.g.\ stopping {after rejecting} a certain subset of hypotheses. Optional stopping contingent on the \emph{final output} is arguably a \emph{much} stronger and more preferable form of selective inference than stopping contingent on some \emph{intermediate} (e.g.\ ``local'' in this case) statistics.

\cref{thm:stop-ebh} indicates that global e-processes lead to stopped FDR control. In fact, we shall establish in \ref{app:comp} the \emph{necessity} of globality for stopped FDR control.
One can show that if the local filtrations are independent, then local e-processes are global e-processes (\ref{app:indp}). However, this rarely happens in reality as it is common to have cross-hypothesis dependence. As \cref{ex:coin} demonstrates, we usually end up constructing local e-processes with no globality guarantee. Therein lies the critical gap between single and multiple hypothesis sequential testing which the upcoming sections aim to address: \emph{global e-processes are necessary for sequential multiple testing, however, we only know how to construct local e-processes.}

\cite{xu2021unified} are aware of this issue, and in their formulation of bandit e-BH framework, they indeed assume that all e-processes are e-processes with respect to a shared global filtration. The same uncareful globality assumption is also made recently by \cite{tavyrikov2025carefree}. It still remains to be answered under what circumstances one may construct e-processes locally within a hypothesis using any of the many existing e-processes in the literature and still enjoy the global property to stop them at a single cross-hypothesis stopping time. 

\section{Global e-processes by a Markovian assumption}\label{sec:mark}

We show in this section that a conditional independence assumption that resembles Markov chains solves the aforementioned filtration issue.
We formulate the sequential multiple testing problem with the following set-up. Consider at time $n=1,2,\dots$, we observe the $n$\textsuperscript{th} observation $Z_n = (X_n, \Y_n)$ that includes the covariate $X_n \in \cX$ and an array of response variables $\Y_n = (Y_n^1,\dots, Y_n^G) \in \cY^1 \times \dots \times \cY^G$. 

We assume that the covariate-response pairs $\{(X_n, \Y_n)\}_{n \ge 1}$ follow a \emph{time-homogeneous marginal conditional model} on $\Y | X$:
\begin{equation}\label{eqn:model}
    Y^g | X \sim p_g( Y^g | X, \theta ), \quad \text{for all } g \in [G],
\end{equation}
which is completely specified by a parameter $\theta \in \Theta$. Here, $p_g(\cdot| x, \theta )$ is a distribution over $\cY^g$ for any $(x, \theta)\in \cX \times \Theta$. That is, denoting the ground truth parameter by $\theta^* \in \Theta$, for any $n \ge 1$ and $g \in [G]$, the conditional distribution of $Y^g_n$ given $X_n$ is $p_g(\cdot| X_n, \theta^* )$.

The model \eqref{eqn:model} is ``marginal'' in the sense that it describes only the marginal distribution of each component $Y^g_n$ among $\Y_n$, given $X_n$. This allows for arbitrary dependence between $Y_n^1, \dots, Y_n^G$ given $X_n$. We further impose the following causal assumption, which becomes crucial in our later discussion on e-processes.



\begin{assumption}\label{asn:causal}
   For any $n \ge 2$, $\Y_{n} \perp ( X_1, \dots, X_{n-1}; \Y_1, \dots, \Y_{n-1}) | X_{n}$.
\end{assumption}
That is, the conditional distribution of $\Y_n | X_n$, whose $G$ marginals are specified by the model \eqref{eqn:model}, is also the conditional distribution of 
\begin{equation}
\Y_n | (X_1, \dots, X_{n}; \Y_1, \dots, \Y_{n-1}).    
\end{equation}
We again stress that
given the covariate $X_n$, there is allowed an arbitrary dependence \emph{between} the $G$ response variables $Y_n^1, \dots, Y_n^G$; but the array of these response variables \emph{as a whole}, $\Y_n$, is independent from everything previously collected. The response $\Y_n$ depends on the past only through its covariate $X_n$, a property akin to that of a Markov chain. In fact, \cref{asn:causal} states that the sequence $
    \{X_1, \Y_1, X_2 ,\Y_2,  X_3 ,\Y_3 \dots \}$
is Markovian at the $\Y_n$'s.
Further, \cref{asn:causal} poses no limitation on the covariate $X_n$. It can be random, deterministic, chosen or sampled adaptively depending on previous observations and inference results. 


We now introduce the $G$ null hypotheses into the model. Consider $G$ subsets of the parameter set, $\Theta_0^g \subseteq \Theta$ for each $g \in [G]$.
The null hypothesis $\cH_0^g$ is true if the ground truth parameter $\theta^* \in \Theta_0^g$.  Recall that in \cref{sec:mt} we used the notation $\cM$ to denote the set of all distributions in the model, and $\cP^g$ all null distributions satisfying the $g$\textsuperscript{th} null hypothesis. The following remark clarifies the consistency from using the previous $\cM$ and $\cP^g$ notations to using the marginal conditional model \eqref{eqn:model}.

\begin{remark} $\cM$ contains all distributions of $\{Z_n\}_{n \ge 1}$ such that (1) \cref{asn:causal} holds; (2) there exists a $\theta^* \in \Theta$ such that, for all $n \ge 1$ and $g \in [G]$, $p_g(\cdot|\cdot, \theta^*)$ is the conditional distribution of 
        $Y_n^g | X_n $. $\cP^g$ contains all such distributions where the $\theta^*$ above can be chosen from $\Theta^g_0$.
\end{remark}
We further remark that,
   while the covariate-response  $(X_n, \Y_n)$ formulation suggests a regression-like set-up where all hypothesis streams receive in synchrony at time $n$ a common covariate $X_n$, our formulation allows more general set-ups. We can remove $X_n$ by simply taking $X_n$ to be a non-random quantity, thus allowing non-regression settings, which include the coin-toss \cref{ex:coin}. The topic of asynchrony is discussed separately in \ref{app:asynch}.

The global filtration we work with differs slightly from the obvious choice $\{\sigma( Z_i :  1 \le  i \le n) \}$ 
but is instead defined as the ``look-ahead'' filtration
 \begin{equation}
         \cF_n = \sigma(   \Y_i, X_j :  1 \le i \le n, 1 \le j \le n+1  ). 
\end{equation}
That is, $\cF_n$ includes all the information after the ($n$+1)\textsuperscript{st} covariate is available but before the ($n$+1)\textsuperscript{st} response variables are revealed. This filtration is one covariate finer than the ``natural'' filtration $ \cG_n = \sigma(   Z_i :  1 \le i \le n)$
which shall also appear in our upcoming theorem.
With these definitions and assumptions above, we establish the following theorem stating that locally defined e-processes, constructed multiplicatively from \emph{stepwise e-values}, can be lifted to this global filtration.

\begin{theorem}\label{thm:glbl}
    Let $g \in [G]$. Suppose for each $n \ge 1$, there is a $\cG_{n-1}$-measurable random function $E_n^g:\Theta^g_0 \times \cX\times \cY^g \to \mathbb R_{\ge 0} $ that satisfies
    \begin{equation}\label{eqn:e-val-property}
      \sup_{\substack{\theta \in \Theta^g_0 \\ x \in \cX} } \int_{\cY^g} E_n^g(\theta , x, y)p_g(dy|x, \theta) \le 1.
    \end{equation}
    Then, under \cref{asn:causal}, under any $\theta^* \in \Theta_0^g$, the process
    \begin{equation}\label{eqn:mtg-theta}
      M_n^g(\theta^*) =   \prod_{i=1}^n E_i^g(\theta^*,X_i, Y_i^g)
    \end{equation}
    is a nonnegative supermartingale on $\{ \cF_n \}$. 
    Consequently, the process
    \begin{equation}\label{eqn:ug}
        U_n^g = \inf_{\theta \in \Theta^g_0} M_n^g(\theta) 
    \end{equation}
    is an e-process on $\{ \cF_n \}$ under the null  $\cH^g_0: \theta^* \in \Theta_0^g$.
\end{theorem}
The proof of the theorem above can be found in \ref{sec:pf-glbl}.
We remark that in the theorem above, the local e-process that tests $\cH^g_0$ is computed via the infimum over a $\cH^g_0$-indexed family of nonnegative supermartingales, each incrementally updated by a stepwise e-value $E_n^g(\theta, X_n, Y_n^g)$ at time $n$ that only takes $X_n$ and $Y_n^g$. The assumption that $E_n^g$ being $\cG_{n-1}$-measurable is only for the sake of generality. In a purely local construction we often let $E_n^g$ be $\sigma(  X_i, Y_i^g : i \le n-1 )$-measurable without looking into other streams, which of course implies $\cG_{n-1}$-measurability.
Recalling the second equivalent \emph{definition} of e-processes in the opening of \cref{sec:seq}, this procedure is universal as it encompasses every possible local construction of e-processes. \cref{thm:glbl} then states that these local e-processes are global.

Additionally, $\{U_n^g\}$ are also e-processes on the natural global filtration $\{ \cG_n \}$ by a simple tower property argument. However, our statement that these are e-processes on $\{ \cF_n \}$ allows stopping decisions after peeking into the upcoming covariates $X_{n+1}$, allowing more flexibility in practice.

Combining with \cref{thm:stop-ebh}, we see that these global e-processes stopped at a global stopping time can produce FDR-controlled multiple testing results with e-BH:

\begin{corollary}
     The set process $\{ \mathsf{eBH}_{\alpha}(U^g_n: g\in [G]  )\}$, where $\{ U^g_n \}$'s are obtained from \eqref{eqn:ug}, satisfies the {level-$\alpha$ stopped FDR control} on $\{ \cF_n \}$. 
\end{corollary}

We provide numerous practical examples in \ref{app:examples}, as well as a counterexample where the failture of \cref{asn:causal} leads to the explicit non-globality of a local e-process. The following biostatistics example, to be elaborated in full in \ref{sec:ui}, may be of particular interest.
\begin{example}[Sequential single-cell differential gene expression testing] \label{ex:scdge}
We test $\cH^g_0: \beta^g = 0$ in the negative binomial generalized linear model
        \begin{equation}\label{eqn:nb}
            Y_n^g \sim \operatorname{NB}( \text{mean}= \exp( X_n \beta^g + \gamma^g ) , \text{dispersion}=\alpha^g )
        \end{equation}
        by the universal inference \citep{wasserman2020universal} e-process 
        \begin{equation}
              U_n^g = \prod_{i=1}^n \frac{ p(Y_i^g | X_i, \hat\beta_{i-1}^g, \hat\gamma_{i-1}^g)  }{ p(Y_i^g | X_i, 0, \hat\gamma_{n}^g) }.
        \end{equation}
        Above, $p(Y_i^g | X_i, \beta^g, \gamma^g)$ is the closed-form probability mass function given by the model \eqref{eqn:nb} (assuming $\alpha^g$ known);  $\hat\beta_{n}^g$, $\hat\gamma_{n}^g$ denote the maximum likelihood estimators of $\beta^g$, $\gamma^g$
        w.r.t.\ the local data stream $\{(X_i, Y_i^g)\}_{1 \le i \le n}$. While this e-process is locally constructed, it can be globalized by
    \cref{asn:causal}, which is reasonable for biostatistics applications as elaborated in \ref{sec:ui}.
\end{example}

Finally, if \cref{asn:causal} fails, one may ``adjust'' a local e-process so it becomes global \citep{choe2024combining}, discussed in \ref{sec:adj}.


\section{Summary}

We propose the stopped e-BH procedure for sequential, anytime-valid multiple hypothesis testing, a procedure both necessary and sufficient for strict false discovery rate control at all stopping times. We carefully distinguish local and global stopping times and e-processes: local e-processes are more commonly constructed, but global (compound) e-processes are more relevant to the FDR control objective. A crucial causal condition is identified that bridges the local-global distinction.
Our work demonstrates the theoretical foundation of using e-values and e-processes for large-scale and complex scientific experiments. These experiments can have arbitrarily dependent hypotheses, sample sizes and sampling schemes for which allowing optional stopping and continuation is necessarily beneficial.

\section*{Acknowledgements}

The authors acknowledge support from NSF grant 2310718.
\bibliography{main}
\newpage
\appendix
\renewcommand{\thesection}{Supp.\arabic{section}}
\section{Ambient filtrations}\label{app:ambient}

In \cref{ex:coin}, we utilize the stardard fact that if $Y_1,Y_2,\dots$ are $[-1,1]$-valued random variables, the ``betting'' process
\begin{equation}
    M_n = \prod_{i=1}^n \left(1 + \frac{Y_i}{2} \right)
\end{equation}
is an e-process for the null
\begin{equation}\label{eqn:correct-null}
    \Exp[ Y_n | Y_1, \dots, Y_{n-1} ] = 0 , \quad\text{for all }n
\end{equation}
on the \emph{natural} filtration
\begin{equation}
    \sigma( Y_1, \dots, Y_n ).
\end{equation}
It is the natural filtration here that leads to the locality of the e-processes in the multiple testing case as described in \cref{ex:coin}.

However, here and in many other single-hypothesis e-process, the filtration need not be natural and can be instead an arbitrary enlargement. For example, if we assume the $[-1,1]$-valued random variables $Y_1, Y_2,\dots$ are adapted to some ambient filtration $\{\cH_n\}$, the very same process $\{M_n\}$ is an e-process on this $\{ \cH_n \}$ for the null
\begin{equation}\label{eqn:wrong-null}
    \Exp[Y_n|\cH_{n-1}] = 0, \quad\text{for all }n.
\end{equation}
If we \emph{take} the ambient filtration $\{\cH_n\}$ to be the global filtration in multiple testing, do we obtain global e-processes for free? However, while the e-processes become global, \emph{they are testing the wrong nulls}! This can be seen from the change from \eqref{eqn:correct-null} to \eqref{eqn:wrong-null} above. We write down now such ``ambient filtration'' variation of \cref{ex:coin} below.
\begin{example}[Multiple coin betting, ambient filtration]\label{ex:coin-amb}
     For each $g \in [G]$, there is a stream $\{Y_n^g\} \iid \operatorname{Rad}(\theta^g)$. Define the global filtration $\cF_n = \sigma( Y_i^g,\dots, Y_n^g  : g \in [G] )$. The process $
         M_n^g = \prod_{i=1}^n \left( 1 + Y^g_i / 2 \right) $
     is an e-process on the global filtration $\{\cF_n \}$ for the null
     \begin{equation}
         \Exp[ Y_n^g | \cF_{n-1}  ]  = 1/2.
     \end{equation}
     Without further causal assumptions, there is \underline{no guanratee} that $ \Exp[ Y_n^g | \cF_{n-1}  ]$ equals $\theta^g$.
 \end{example}
 Indeed, it can be easily seen that should \cref{asn:causal} be met, $(Y_n^g : g \in [G])$ is independent from the past $\cF_{n-1}$, then $ \Exp[ Y_n^g | \cF_{n-1}  ]$ equals $\theta^g$. In summary, using the global filtration as the ambient filtration does not avoid the issue of globality.

\section{Asynchrony}\label{app:asynch}

We now discuss the situation where data streams of different hypotheses arrive asynchronously. We continue to work with the Rademacher coin toss model as \cref{ex:coin}.
Let $ \{Y^1_n \} \iid \rad(\theta^1) $ be a ``fast'' stream and $ \{Y^2_n \} \iid \rad(\theta^2)$ a slow stream, such that $Y^1_3$ may exert a causal effect on $Y_2^2$ (say $Y^1_3 = Y^2_2$). For example, assume the temporal order of observations follows \cref{tab:arriv}.
\begin{table}[!h]
    \centering
    \begin{tabular}{c|c|c|c|c|c}
        $Y^1_1$  &  $Y_2^1$ &   $Y_3^1$ &  $Y_4^1$ & $Y_5^1$ & $\dots$  \\ \hline
        &  $Y^2_1$ &&   & $Y^2_2$ & $\dots$
    \end{tabular}
    \caption{Order of data arrival.}
    \label{tab:arriv}
\end{table}
This setting is then similar to our counterexample in \cref{sec:counterexample} where one stream foretells the other (here $Y_2^2$ ``foretells'' $Y_3^1$, understood index-wise), and all the consequences apply, if one aligns the observation counts by  ``waiting'' for the slow stream and erroneously applying stopped e-BH to the local e-processes
\begin{equation}
    M^1_n = \prod_{i=1}^n \left(1 + \frac{Y^1_i}{2} \right), \quad
    M^2_n = \prod_{i=1}^n \left(1 + \frac{Y^2_i}{2} \right).
\end{equation}

The correct approach here is to use the natural time instead, grouping and reindexing the sample. For example, with the data arrival scheme \cref{tab:arriv},
define two new streams that are synchronized:
\begin{gather}
    W^1_1 = (Y^1_1, Y^1_2), \quad W^1_2 = (Y^1_3, Y^1_4, Y^1_5), \quad \dots \\
    W^2_1=Y^2_1, \quad W^2_2=Y^2_2, \quad \dots
\end{gather}
Then \cref{asn:causal} requires the independence
\begin{equation}\label{eqn:indep}
    \mathbf W_n \perp (\mathbf W_1,\dots, \mathbf W_{n-1}  ).
\end{equation}
Once \eqref{eqn:indep} holds, one may now apply stopped e-BH on the reindexed local e-processes
\begin{gather}
   {\widetilde M}^1_1 = M^1_2, \quad {\widetilde M}^1_2 = M^1_5, \quad \dots \\
    {\widetilde M}^2_1=M^2_1, \quad {\widetilde M}^2_2=M^2_2, \quad \dots
\end{gather}


Note that in this setting we can allow causality from  $Y^1_3$ to $Y_2^2$ (now as cross-hypthesis dependence within simultaneous observation), but not from $Y^1_2$ to $Y_2^2$ (as this would violate \eqref{eqn:indep}). Indeed, after time synchronization, one still needs to make sure \cref{asn:causal} holds, as is the case with synchronous streams.

\section{Compound e-processes and universality}\label{app:comp}

The e-BH procedure can be applied to random variables satisfying a condition weaker than being e-values \citep[Proposition 3]{wang2022false} while still holding the FDR control. 
Such random variables are named \emph{compound e-values} by \citet[Definition 1.1]{ignatiadis2024compound}.

\begin{definition}[Compound e-values]\label{def:comp-e}
    $G$ random variables $\mathbf{E} = (E_g : g \in [G]) \in \mathbb R_{\ge 0 }^G $ are called \emph{compound e-values for $(\cP^g : g \in G)$} if,
    \begin{equation}
      \sup_{P \in \cM}  \sum_{g\in [G]} \id\{ P \in \cP^g \} \cdot \Exp_{P} E_g \le G.
    \end{equation}
\end{definition}
Clearly, if each $E_g$ is an e-value for $\cP^g$, that is, $\sup_{P \in \cP^g}\Exp_{P} E_g \le 1$, then $\mathbf{E}$ are compound e-values. However compound e-values form a much larger class including weighted e-values. The e-BH on compound e-values controls the FDR due to \citet[Theorem 3.2]{ignatiadis2024compound}.

\begin{lemma}[FDR property of compound e-BH]\label{lem:fdr-ebh-comp} Let $\mathbf{E} = (E_g : g \in [G]) $ be compound e-values for $(\cP^g: g\in [G])$. Then, for any underlying {$P \in \cM$},
\begin{equation}
  \fdr_P[\mathsf{eBH}_{\alpha}(\mathbf{E})]  \le \alpha.
\end{equation}
\end{lemma}

It is worth noting that the converse to \cref{lem:fdr-ebh-comp} is also true, due to \citet[Theorems 3.1 and 3.3]{ignatiadis2024compound}. That is, if a random set $R \subseteq [G]$ whose FDR is below $\alpha$, there exist $G$ compound e-values $\mathbf E$ such that $ \mathsf{eBH}_{\alpha}(\mathbf E)  =R$. 

We now spell out the sequential counterparts of these results.
Below, we denote the set of all stopping times on $\{\cF^g_n \}$ by $\ST^g$, and on $\{ \cF_n \}$ by $\ST$, and define local and global compound e-processes.

\begin{definition}\label{def:processes}
    Let $(\{ M^g_n \}: g\in [G])$ be $G$ nonnegative processes. They are called:
    \begin{enumerate}
        \item \emph{local e-processes for $(\cP^g : g \in [G])$}, if for any $g \in [G]$, $\{ M^g_n \}$ is adapted to $\{  \cF_n^g \}$ and
        \begin{equation}
            \sup_{ \substack{ P \in \cP^g \\ \tau \in \ST^g } } \Exp_{P} M^g_{\tau} \le 1;
        \end{equation}
         \item \emph{local compound e-processes for $(\cP^g : g \in [G])$}, if for any $g \in [G]$, $\{ M^g_n \}$ is adapted to $\{  \cF_n^g \}$ and
        \begin{equation}
            \sup_{ \substack{  P \in \cM \\ \tau^1 \in \ST^1 \\ \dots \\ \tau^G \in \ST^G  } } \left( \sum_{g\in[G]} \id\{ P \in \cP^g \} \cdot \Exp_{P} M^g_{\tau^g} \right) \le G;
        \end{equation}
         \item \emph{global e-processes for $(\cP^g : g \in [G])$}, if for any $g \in [G]$, $\{ M^g_n \}$ is adapted to $\{  \cF_n \}$ and
        \begin{equation}
            \sup_{ \substack{ P \in \cP^g \\ \tau \in \ST } } \Exp_{P} M^g_{\tau} \le 1;
        \end{equation}
         \item \emph{global compound e-processes for $(\cP^g : g \in [G])$}, if for any $g \in [G]$, $\{ M^g_n \}$ is adapted to $\{  \cF_n \}$ and
        \begin{equation}
            \sup_{ \substack{  P \in \cM \\ \tau \in \ST } } \left( \sum_{g\in[G]}\id\{ 
P \in \cP^g  \} \cdot\Exp_{P} M^g_{\tau} \right) \le G.
        \end{equation}
    \end{enumerate}
\end{definition}

These definitions satisfy the clear inclusion relations $\textit{1} \subseteq \textit{2}$ and $\textit{3} \subseteq \textit{4}$. $\textit{1}$ and $\textit{3}$, for example, are not included in either direction, because while adaptivity to a larger filtration is a weaker condition, returning an e-value at any stopping time on a larger filtration is a stronger condition. Nonetheless, local and local compound e-processes are much more natural and straightforward to construct, as we have previously argued in \cref{ex:coin}. The following statement extends \cref{thm:stop-ebh}, whose proof is straightforward.

\begin{proposition}\label{thm:stop-ebh-comp}
    Let $(\{ M^g_n \}: g\in [G])$ be global compound e-processes. Then, the set process $\{ \mathsf{eBH}_{\alpha}(M^g_n: g\in [G]  )\}$  satisfies the {level-$\alpha$ stopped FDR control}.
\end{proposition}

We now establish the following converse to \cref{thm:stop-ebh-comp}.

\begin{theorem}[Universality of stopped e-BH]
     Let  $\{ R_n \}$ be a set process adapted to $\{\cF_{n}\}$ satisfiying the level-$\alpha$ stopped FDR control. Then, there exist $G$ global compound e-processes $(\{ M^g_n \}: g\in [G])$ for $(\cP^g : g \in [G])$ such that
     \begin{equation}
        R_n = \mathsf{eBH}_{\alpha}(M^g_n: g\in [G]  ) \quad \text{for all }n.
    \end{equation}
\end{theorem}
\begin{proof}
    First, at any non-random $n$, since $R_n$ is a $\cF_n$-measurable random set that controls the FDR at level $\alpha$, it follows from \citet[Theorems 3.1 and 3.3]{ignatiadis2024compound} that
    \begin{equation}
        R_n = \mathsf{eBH}_{\alpha}(M^g_n: g\in [G]  )
    \end{equation}
    for $G$ compound e-\emph{values} $M^1_n,\dots, M^G_n$ {defined by}
    \begin{equation}
        M^g_n = \frac{G}{\alpha} \cdot \frac{\id\{ g \in R_n \}}{ |R_n| \vee 1 }.
    \end{equation}
    Hence {these compound e-values} are $\cF_n$-measurable as well. We have thus defined the processes $\{M_n^g\}$ for $g \in [G]$, all globally adapted, and it remains to verify they are {compound} e-processes globally. Take any $\tau \in \ST$. Since the FDR is controlled for the set $R_\tau$, again we have
    \begin{equation}
        R_\tau = \mathsf{eBH}_{\alpha}(E^g_\tau: g\in [G]  )
    \end{equation}
    for $G$ compound e-{values} $E^1_\tau,\dots, E^G_\tau$ where
    \begin{equation}
         E^g_\tau = \frac{G}{\alpha} \cdot \frac{\id\{ g \in R_\tau \}}{ |R_\tau| \vee 1 },
    \end{equation}
    which equals $M^g_{\tau}$. Therefore 
    \begin{equation}
        \sum_{g\in [G]} \id\{ P \in \cP^g \} \Exp_{P} M^g_\tau =  \sum_{g\in [G]} \id\{ P \in \cP^g \} \Exp_{P} E^g_\tau  \le G
    \end{equation}
    under any {$P \in \cM$},
    concluding that they are global compound e-processes.
\end{proof}

\section{Global e-processes via adjusters}\label{sec:adj}

While we have discussed in \cref{sec:mark} certain conditions under which local e-processes are naturally global e-processes, we now quote a complementary recent result by \cite{choe2024combining} which states that by slightly ``muting'' an e-process, it becomes an e-process on arbitrary refinement  of the filtration. This is helpful when \cref{asn:causal} does not hold or is not verifiable.
We first recall the definition of an \emph{adjuster}, which traces back to \cite{shafer2011test}.

\begin{definition}\label{def:adj}
 An adjuster is a non-decreasing function $\mathsf{A}: \mathbb R_{\ge 1} \to \mathbb R_{\ge 0} $ such that
 \begin{equation}
     \int_1^\infty \frac{\mathsf A(x)}{x^2} dx \le 1.
 \end{equation}
\end{definition}
Simple examples include $x \mapsto kx^{1-k}$ for $k\in(0,1)$ and $\sqrt{x} -1$. Applying an adjuster to an e-process, or more generally the running maximum of an e-process is referred to as ``e-lifting'' by \citet[Theorem 2]{choe2024combining} and yields an e-process for any finer filtrations. That is,
    if $\{ M_n \}$ is an e-process for $\cP$ on $\{ \cG_n \}$, then for any adjuster $\mathsf A$ and any filtration $\{ \cF_n \} \supseteq \{ \cG_n \}$, the adjusted process $\{ M_n^a \}$ where
    \begin{equation}
       M_0^a = 1, \quad M_n^a = \mathsf{A}\left( \max_{0\le i \le n}M_i \right)
    \end{equation}
    is an e-process for $\cP$ on $\{ \cF_n \}$.
Following the discussion on compound e-processes in \cref{app:comp}, we formalize the concept of \emph{compound adjusters} below for the sake of generality.

\begin{definition}
     Let $( \mathsf{A}^g  : g \in [G] )$ be a family of non-decreasing functions $\mathbb R_{\ge 1} \to \mathbb R_{\ge 0} $. They are called \emph{compound adjusters} if their average is an adjuster, i.e.,
     \begin{equation}
     \int_1^\infty \frac{\sum_{g \in [G]}\mathsf A^g(x)}{x^2} dx \le G.
 \end{equation}
\end{definition}

That is, compound adjusters arise from decomposing an adjuster into $G$ monotone components.
For example, a vector of $G$ adjusters are compound adjusters; and similar to compound e-values, one can construct compound adjusters by taking a weighted sum of $G$ adjusters as long as weights sum up $\le G$. The following statement shows that we can first apply compound adjusters then apply e-BH to local e-processes to control the stopped FDR.

\begin{corollary}\label{cor:adj-comp}
 Let $(\{ M_n^g \} : g \in [G])$ be local e-processes and  $( \mathsf{A}^g  : g \in [G] )$ be compound adjusters.  Then, the set process $\{ \mathsf{eBH}_{\alpha}(M^{ga}_n: g\in [G]  )\}$  satisfies the level-$\alpha$ stopped FDR control, where 
  \begin{equation}
       M_0^{ga} = 1, \quad M_n^{ga} = \mathsf{A}^g\left( \max_{0\le i \le n}M_i^g \right).
    \end{equation}
\end{corollary}
\begin{proof}
   Since $\{  M_n^{ga} \}$ are global compound e-processes, this follows from \cref{thm:stop-ebh-comp}. 
\end{proof}

\section{Examples {and counterexamples}}\label{app:examples}

\subsection{Multivariate Z-test}

We present a simple example of \cref{thm:glbl} with dependence across hypotheses via the following multivariate Z-test set-up. Let the covariates $X_n \in \cX$ be non-existent (or constant) and responses $Y^g_n 
 \in \cY^g =  \mathbb R$ for all $g \in [G]$. The joint response vectors $\Y_n = (Y_n^g : g \in [G])$ are drawn i.i.d.\ from a multivariate normal distribution
 \begin{equation}
     \Y_1, \Y_2, \dots \iid \mathcal{N}(\theta, \Sigma),
 \end{equation}
 where the mean vector $\theta \in \Theta = \mathbb R^G$ and the off-diagonal entries of the covariance matrix $\Sigma$ are unknown, while the diagonal entries of $\Sigma$ is known in advance. 
 We test the null hypotheses
 \begin{equation}
     \cH_0^g: \Exp[ Y_n^g  ] = \theta^g = 0.
 \end{equation}
That is, the marginal conditional model \eqref{eqn:model} is furnished with distributions $p_g( \cdot | x, \theta) = \mathcal{N}(\theta^g, \Sigma_{gg})$, and the null sets
$\Theta^g_0 = \{ \theta \in \mathbb R^G : \theta^g = 0 \}$. The model clearly meets \cref{asn:causal}. Notably, the off-diagonal entries of $\Sigma$ induce instantaneous correlation across the $G$ coordinates. However, as we shall see, only the diagonal entries of $\Sigma$ are used in the construction of the e-processes. We define the stepwise e-value function as
\begin{equation}
    E_n^g(\theta, x, y) = \exp\left( \eta_n^g y - \frac{1}{2}  \eta_n^{g2} \Sigma_{gg} \right),
\end{equation}
where $\eta_n^g$ is a $\cG_{n-1}$-measurable ``learning rate'' parameter. By a direct calculation, \eqref{eqn:e-val-property} holds with equality. Note that the  stepwise e-value function $E_n^g(\theta, x, y) $ does \emph{not} dependent on the parameter $\theta$. Therefore, the process $\{M^g_n(\theta)\}$ defined in \eqref{eqn:mtg-theta} which is a global martingale under the ground truth $\theta^* = \theta \in \Theta_0^g$, is the same regardless of $\theta \in \Theta_0^g$. The corresponding e-process is therefore this martingale $M^g_n(\theta)$ with any $\theta \in \Theta_0^g$:
\begin{equation}
    U_n^g = M_n^g(\theta) = \exp \left\{  \sum_{i=1}^n \eta_i^g Y_i^g - \frac{ \Sigma_{gg} \sum_{i=1}^n \eta_i^{g2} }{2}   \right\}.
\end{equation}

\subsection{Multiple sequential probability ratio test}

The Gaussian example above is a special case of the simple-versus-simple sequential probability ratio test (SPRT). Let us spell out the general SPRT setup. 
Here, we allow arbitrary $\cX$ and $\cY^g$ (on which we equip a base measure and simply write it as $dy$)  but assume that the nulls are ``simple'' in the following sense:  for each $g\in [G]$, all null parameters $\theta \in \Theta^g_0$ share the same conditional probability density or mass function $p_g(\cdot | x, \theta)$ in the model \eqref{eqn:model}, denoted by $p_g(\cdot | x)$. We further specify an alternative conditional probability density or mass function $q_g(\cdot | x)$ for each $g \in [G]$. We allow the Markov-like condition \cref{asn:causal} satisfied with arbitrary cross-hypothesis dependence among $\Y | X$. Then, the stepwise e-value function $E_n^g$ takes the likelihood ratio
\begin{equation}
    E_n^g(\theta, x, y) = \frac{ q_g(y |x)  }{ p_g(y | x) },
\end{equation}
Once again, \eqref{eqn:e-val-property} holds with equality and  $E_n^g(\theta, x, y) $ does not dependent on the parameter $\theta$. The corresponding e-process therefore equals the global test martingale
$M^g_n(\theta)$ with any $\theta \in \Theta_0^g$:
\begin{equation}
    U_n^g = M_n^g(\theta) = \prod_{i=1}^n \frac{ q_g(Y_i^g |X_i)  }{ p_g( Y_i^g |  X_i ) }.
\end{equation}

\subsection{Parametric regression with universal inference}\label{sec:ui}

Even more generally, \cref{thm:glbl} applies as long as we work with a parametric model with computable and optimizable likelihood functions. Nulls need not be simple and alternatives can evolve over time to approximate the ground truth. In this case, the e-processes can be computed with the \emph{universal inference} method
 due to \cite{wasserman2020universal}. The example again allows the Markov-like condition \cref{asn:causal} satisfied with arbitrary cross-hypothesis dependence among $\Y | X$.

We assume the global parameter $\theta$ can be split into $G$ $\mathbb R^m$-valued, local parameters $(\theta^g = \phi^g(\theta): g \in [G])$ via $G$ functions $\phi^g : \Theta \to \mathbb R^m$, with the $g$\textsuperscript{th} null region being $\Phi_0^g \subseteq \mathbb R^m$; that is,
\begin{equation}
    \Theta^g_0 = \{ \theta \in \Theta : \phi^g(\theta) \in \Phi_0^g \}.
\end{equation}
For example, $\Theta$ can be $\mathbb R^{m\times G}$ and $\phi^g(\theta)$ can be the $g$\textsuperscript{th} column of the matrix $\theta$. But our assumption above is more flexible.
With a slight abuse of notations, we assume the model \eqref{eqn:model} is specified by the conditional probability density or mass functions $p_g(y|x,\theta^g)$. Then, the stepwise e-value function $E_n^g$ of universal inference is defined as the likelihood ratio
\begin{equation}
    E_n^g(\theta, x, y) = \frac{ p_g(y |x, \tilde\theta_{n-1}^g)  }{ p_g(y | x, \theta^g) },
\end{equation}
where $\tilde\theta_{n-1}^g$ is any estimator for $\theta^g$ computed from $Z_1,\dots, Z_{n-1}$, thus is $\cG_{n-1}$-measurable.   To see that \eqref{eqn:e-val-property} holds with equality again, for any $\theta \in \Theta_0^g$,
 \begin{equation}
    \int_{\cY^g} E_n^g(\theta , x, y)p_g(y|x, \theta^g) dy =  \int_{\cY^g}   p_g(y |x, \tilde\theta_{n-1}^g) dy = 1.
\end{equation}
In particular, define the maximum likelihood estimates within the null and the full models
\begin{equation}\label{eqn:mle}
    \hat\theta_n^g =  \argmax_{ \theta \in \Theta_0^g } \prod_{i=1}^n p_g(Y_i^g | X_i , \theta), \quad \tilde\theta_n^g =  \argmax_{ \theta \in \Theta^g } \prod_{i=1}^n p_g(Y_i^g | X_i , \theta).
\end{equation}
Then the e-processes $\{ U_n^g \}$ take the following simple form
\begin{equation}
    U_n^g = \prod_{i=1}^n \frac{ p_g(Y_i^g | X_i, \tilde\theta_{i-1}^g)  }{ p_g(Y_i^g | X_i, \hat\theta_{n}^g) }.
\end{equation}
Numerous generalized regression-like statistical models take this form, including the single-cell differential gene expression testing problem we mentioned as \cref{ex:scdge}. 
Let $Y_n^g$ be the RNA-seq count (i.e.\ gene expression level) of gene $g$ and cell $n$. Let $X_n \in \{0,1\}$ be the group membership (disease or control) of this cell. The counts follow the
negative binomial generalized linear model
        \begin{equation}
            Y_n^g \sim \operatorname{NB}( \text{mean}= \exp( X_n \beta^g + \gamma^g ) , \text{dispersion}=\alpha^g ).
        \end{equation}
The null hypothesis $\cH^g_0$ that ``gene $g$ is not differentially expressed between the two groups'' asserts that, in the model above,
        the coefficient $\beta^g$ equals 0. 
        
        One can then construct the universal inference e-process for $\cH_0^g$ via the aforementioned scheme with $\theta = ((\beta^g,\gamma^g) : g \in [G])$ and $\Theta_0^g = \{  \theta : \beta^g=0  \}$. The process $\{U_n^g\}$ can readily be written as 
    \begin{equation}
              U_n^g = \prod_{i=1}^n \frac{ \operatorname{NB}( Y_i^g | \text{mean}= \exp(X_i  \hat\beta_{i-1}^g + \hat\gamma_{i-1}^g ) , \text{dispersion}=\alpha^g ) }{\operatorname{NB}(Y_i^g |  \text{mean}= \exp( \hat\gamma_{n}^g ) , \text{dispersion}=\alpha^g )}.
        \end{equation}
        For this to be a global e-process, 
    \cref{asn:causal} states that each cell's counts are causally related to the past only through the cell's own group membership indicator, a reasonable assumption for the generalized linear model.


{

\subsection{Nonparametric heavy-tailed conditional mean testing}

Our final example of \cref{thm:glbl} is nonparametric, where $\cY^g = \mathbb R$ and for each $g \in [G]$ we specify a function $\mu^g: \cX \to \mathbb R$. We test the hypotheses
\begin{equation}
    \cH^g_0: \Exp[Y^g | X] \le \mu^g(X)
\end{equation}
under the finite conditional variance assumption $\Var[Y^g | X] \le v^{g}(X)$. Therefore, the parameter set {$\Theta$ indexes \emph{all} conditional laws $\Y|X$ such that \newline $\sup_{g \in [G]}( \Var[Y^g | X] -  v^{g}(X))\le 0$}, and $\Theta^g_0$ all those in $\Theta$ with $\Exp[Y^g | X] \le \mu^g(X)$.  {Again, \cref{asn:causal} is satisfied with arbitrary cross-hypothesis dependence among $\Y | X$.} We employ the following ``Catoni-style'' e-value due to \cite{WANG2023168}:
\begin{equation}
    E_n^g(\theta , x, y )= \exp( \phi(\lambda_n(y - \mu^g(x))) - \lambda_n^2 v^{g}(x)/2)
\end{equation}
where
\begin{equation}
    \phi (x) = \begin{cases} \log(1 + x + x^2/2), & x \ge 0, \\ -\log(1 - x + x^2/2), & x < 0, \end{cases}
\end{equation}
and $\lambda_n > 0$ is $\cG_{n-1}$-measurable. Again, the function $E_n^g$ does {not} dependent on the ``parameter'' $\theta$ that now lives in an infinite-dimensional space. The corresponding e-process therefore equals the supermartingale $M^g_n(\theta)$ with any $\theta \in \Theta_0^g$:
\begin{equation}
    U_n^g = M_n^g(\theta) = \exp \left\{  \sum_{i=1}^n \phi (   \lambda_i(   Y_i^g - \mu^g(X_i) ) ) - \frac{\sum_{i=1}^n \lambda_i^2 v^{g}(X_i)}{2}   \right\}.
\end{equation}
If the nulls are $\Exp[Y^g | X] \ge \mu^g(X)$ instead, one can simply replace $\phi$ with $-\phi$; and if they are $\Exp[Y^g | X] = \mu^g(X)$, use a convex combination (e.g.\ the average) of the e-process for $\Exp[Y^g | X] \le \mu^g(X)$ and the e-process for $\Exp[Y^g | X] \ge \mu^g(X)$. This example thus encompasses a wide scope of nonparametric problems including heavy-tailed linear regression.

\subsection{Failure of \cref{asn:causal}}
\label{sec:counterexample}
We now illustrate a simple counterexample where the Markovian causal condition (\cref{asn:causal}) does not hold. Our example is akin to a recent construction by \citet[Section 9.2]{dandapanthula2025multiple}, and builds on the coin-toss setting of \cref{ex:coin}. This involves $G=2$ streams of variables where the second stream ``peeks into the future'' of the first stream.

We construct two different set-ups distinguished by $d\in\{0,1\}$. Consider empty covariates $\{ X_n \}$, and the response variables for $\cH_0^1$, $\{Y^1_n\}$, an i.i.d.\ sequence of Rademacher random variables with $\Pr[Y^1 = 1] = \theta$ and $\Pr[Y^1 = - 1] = 1 -\theta$. Now, we define $Y^2_n = Y^1_{n+d}$ for all $n \ge 1$, where we recall $d$ is either 0 or 1. 
That is, the second stream, while also an i.i.d.\ Rademacher sequence, reproduces the coin tosses of the first stream if $d=0$; and foretells the upcoming coin tosses of the first stream if $d=1$. 
The statistician tests $\cH^1_0:\theta = 1/2$ using $\{Y^1_n\}$, and $\cH^2_0:\theta = 1/2$ using $\{Y^2_n\}$ 
by applying the stopped e-BH to the two local e-processes (as \cref{ex:coin})
\begin{equation}
    M^1_n = \prod_{i=1}^n \left(1 + \frac{Y^1_i}{2} \right), \quad
    M^2_n = \prod_{i=1}^n \left(1 + \frac{Y^2_i}{2} \right),
\end{equation}
with respect to the local filtrations
\begin{equation}
    \cF^1_n = \sigma(Y^1_1, \dots, Y^1_n), \quad   \cF^2_n = \sigma(Y^2_1, \dots, Y^2_n).
\end{equation}
Here, note that the global filtration $\cF_n = \cF^1_{n+d}$.

If $d=0$, the test is perfectly valid as \cref{asn:causal}, equivalent to
\begin{equation}
    Y^1_n \perp ( Y^1_1, \dots,  Y^1_{n-1}),
\end{equation}
holds due to the i.i.d.\ assumption. 

However, if $d=1$, \cref{asn:causal} is now equivalent to
\begin{equation}
    (Y^1_n, Y^1_{n+1} ) \perp  (Y^1_1, \dots,  Y^1_{n}),
\end{equation}
and is therefore not satisfied. In this case, the following global stopping time breaks the first local e-process:
\begin{equation}
    \tau = 1 + \id\{ Y^2_1  = 1 \} =  1 + \id\{ Y^1_2  = 1 \}.
\end{equation}
To see that it is a stopping time with respect to $\{\cF_n\}$, $\{ \tau = 1 \} = \{ Y^2_1 = -1 \} \in \cF_1$. However, the calculation
\begin{center}
    \begin{tabular}{|c|c|c|c|c|c|}  \hline
    $Y^1_1$ &  $M^1_1$ & $Y^1_2$ &  $M^1_2$  & $\tau$ & $M^1_{\tau}$ \\ \hline
   \multirow{2}{*}{1}  & \multirow{2}{*}{1.5} & 1 & 2.25 & 2 & 2.25  \\ 
   & & -1 & 0.75 & 1 & 1.5 \\
   \multirow{2}{*}{-1} &  \multirow{2}{*}{0.5} & 1 & 0.75 & 2 & 0.75 \\
   & & -1 & 0.25 & 1 & 0.5 \\ \hline 
\end{tabular}
\end{center}
shows that $\Exp_{\theta=1/2}(M^1_{\tau}) = 1.25 > 1$. The local e-process $\{ M^1_n \}$, therefore, is not a global e-process. Intuitively, a gambler betting heads on the coin tosses $\{ Y^1_n \}$ peeks into the future after one round of bet, only taking the next bet if it will be profitable.

It is worth remarking here, however, that the globally stopped local e-process $M^1_\tau$, while not an e-value, still satisfies $\Pr_{\theta = 1/2}(M^1_\tau \ge 1/\alpha)\le \alpha$ for any $\alpha \in (0,1)$, i.e.\ Markov's inequality  \emph{as if $\Exp_{\theta = 1/2} M^1_\tau $ were 1}. See \citet[Lemma 3]{howardcs}. That is, $1/M^1_\tau$ is a valid p-value. Recall that the e-BH procedure is defined as the BH procedure with inverted inputs \citep[Section 4.1]{wang2022false},
\begin{equation}
     \mathsf{eBH}_{\alpha}(E_1,\dots, E_G) =  \mathsf{BH}_{\alpha}(E_1^{-1},\dots, E_G^{-1}).
\end{equation}
Therefore, if one applies the stopped e-BH procedure to local e-processes at a global stopping time $\tau$,
\begin{equation}
     \mathsf{eBH}_{\alpha}(M^g_\tau : g \in [G]),
\end{equation}
one is essentially applying the BH procedure to $G$ p-values,
\begin{equation}
     \mathsf{BH}_{\alpha}(1/M^g_\tau : g \in [G]),
\end{equation}
therefore potentially facing the same possible FDR inflation that can at worst be $1 + 2^{-1} + \dots + G^{-1} \approx \log G$, as the one arising in applying BH to arbitrarily dependent p-values \citep[Theorem 1]{wang2022false}.

The complication above does not happen as long as the time index $n$ is a bona fide representation of the chronological evolution of the experiment, and no retrocausality is allowed. 
Local e-processes may otherwise fail to be global e-processes.
A remedy, we note, is the filtration-agnostic adjustment by \cite{choe2024combining} discussed in \cref{sec:adj}.
}

\section{Omitted proofs}

\subsection{Proof of \cref{thm:glbl}}\label{sec:pf-glbl}
\begin{proof} Suppose the ground truth $\theta^* \in \Theta^g_0$. Since the function $E_n^g$ is $\cG_{n-1}$-measurable, there exists a \emph{non-random} function $f_E$ such that
    \begin{equation}
         E_n^g(\theta,X_n, Y_n^g) = f_E(\theta,X_n, Y_n^g, Z^{n-1})
    \end{equation}
    where $Z^{n-1} = (Z_i :1 \le i \le n-1)$. Now we have
    \begin{align}
        & \Exp_{\theta^*}[ E_n^g(\theta^*,X_n, Y_n^g) | \cF_{n-1}]
        \\
        = & \Exp_{\theta^*}[  f_E(\theta^*,X_n, Y_n^g, Z^{n-1}) |   X_n, Z^{n-1}   ]
        \\
          = & \int_{\cY^g} f_E(\theta^*,X_n, y, Z^{n-1}) p_g(dy|X_n, \theta^*) 
        \\
        = & \int_{\cY^g} E_n^g(\theta^*,X_n, y) p_g(dy|X_n, \theta^*) \le 1.
    \end{align}
    In the calculation above, the Markovian \cref{asn:causal} is implicitly invoked at the second equality. On the one hand, $p_g(\cdot | \cdot,  \theta^*)$ specifies the conditional distribution of $Y^g_n | X_n$. On the other hand,  \cref{asn:causal} implies $Y^g_n \perp Z^{n-1} | X_n$, therefore $Y^g_n | X_n, Z^{n-1}$ has the same conditional distribution as $Y^g_n | X_n$, which is $p_g(\cdot | \cdot,  \theta^*)$. This makes the second equality valid.

    We have shown that under $\theta^*$, the conditional expectation given $\cF_{n-1}$ of $E_n^g(\theta^*,X_n, Y_n^g)$ is at most 1. Since $E_n^g(\theta^*,X_n, Y_n^g)$ is itself $\cF_n$-measurable, the product $M_n^g(\theta^*)$ is therefore a supermartingale on $\{\cF_n \}$. This concludes the proof.
\end{proof}

\subsection{Independent local filtrations}\label{app:indp}

In this section, we formally prove that when there is no cross-hypothesis dependence, local e-processes are always global e-processes.

\begin{proposition}\label{prop:indp}
    Following the terminology in \cref{def:processes}, if the local $\sigma$-algebras at infinity $( \cF_\infty^g  : g \in [G])$ are independent, local e-processes $(\{ M_n^g \} : g \in [G])$ are global.
\end{proposition}

\begin{proof} Our proof is based on the two equivalent definitions of e-processes mentioned in \cref{sec:seq}.
    Fix a $g \in [G]$. Since $\{ M_n^g \} $ is an e-process for $\cP^g$ on $\{ \cF^g_n \}$, for every $P \in \cP^g$, it is upper bounded by a nonnegative supermartingale $\{N_n^{P}\}$ on $\{ \cF^g_n \}$ with  $N_0^{P} = 1$. It suffices to prove it is a supermartingale on $\{ \cF_n \}$ as well. This follows from
    \begin{equation}
        \Exp[ N_{n+1}^{P}  |  \cF_n ] =  \Exp[ N_{n+1}^{P}  |  \cF_n^{g}; \cF_n^{h} : h \in [G] \setminus \{ g \} ]  \stackrel{(*)}{=}   \Exp[ N_{n+1}^{P}  |  \cF_n^{g}] \le N_n,
    \end{equation}
    where the equality $(*)$ follows from lifting \cref{lem:3events} to $\sigma$-algebras: $\sigma(\cF_n^{h} : h \in [G] \setminus \{ g \})$ being independent from $\sigma( N_{n+1}^P, \cF_n^g  )$.
\end{proof}

\begin{lemma}\label{lem:3events}
    Let $A, B, C$ be events in a probability space such that $B$ is independent from $\sigma(A, C)$, and $\Pr[B], \Pr[C] > 0$. Then $\Pr[A| B \cap C] = \Pr[A|C]$.
\end{lemma}

\begin{proof}
    $\Pr[A| B \cap C] = \frac{\Pr[A\cap B \cap C]}{\Pr[B \cap C]} = \frac{\Pr[A \cap C] \Pr[B]}{\Pr[B] \Pr[C]} =  \frac{\Pr[A \cap C] }{\Pr[C]} = \Pr[A|C]$.
\end{proof}

\end{document}